\documentclass[twoside,a4paper,11pt]{amsart}

\usepackage[english]{babel}
\usepackage{amsmath}
\usepackage{amssymb}
\usepackage{url}
\usepackage{amsthm}
\usepackage{amsxtra}
\usepackage{mathrsfs}
\usepackage{fancyhdr}
\newtheorem{theorem}{Theorem}
\newtheorem{lemma}{Lemma}[section]
\newtheorem{proposition}[lemma]{Proposition}

\theoremstyle{definition}

\numberwithin{equation}{section}

\newcommand{\beq}{\[}
\newcommand{\eeq}{\]}
\newcommand{\RE}{\mathbb R}

\newcommand{\R}{\mathbb R}

\newcommand{\C}{\mathbb C}

\newcommand{\E}{\mathbb E}
\newcommand{\PP}{\mathbb P}

\newcommand{\NN}{\mathcal{N}}

\newcommand{\Tr}{\operatorname{Tr}}

\newcommand{\ve}{\varepsilon}
\newcommand{\al}{\alpha}

\newcommand{\de}{\delta}

\renewcommand{\Im}{\operatorname{Im}}
\renewcommand{\Re}{\operatorname{Re}}

\renewcommand{\ln}{\log}

\newcommand{\uu}{{\bf u}}
\newcommand{\vv}{{\bf v}}
\newcommand{\xx}{{\bf x}}

\newcommand{\ww}{{\bf w}}

\newcommand{\s}{s}

\newcommand{\eps}{\varepsilon}

\newcommand{\NH}{\mathcal{N}}

\newcommand{\hscale}{\eta}

\title[]{Local Marchenko-Pastur Law at the Hard Edge of Sample Covariance Matrices}

\author[]{Claudio Cacciapuoti\email{cacciapuoti@hcm.uni-bonn.de} \and Anna Maltsev\email{annavmaltsev@gmail.com}  \and Benjamin Schlein
\email{benjamin.schlein@hcm.uni-bonn.de}}
\address{Hausdorff Center for Mathematics\\
Institute for Applied Mathematics, University of Bonn\\
Endenicher Allee 60, 53115 Bonn, Germany}

\begin{document}

\maketitle

\begin{abstract}
Let $X_N$ be a $N\times N$ matrix whose entries are i.i.d. complex random variables with mean zero and variance $\frac{1}{N}$. We study the asymptotic spectral distribution of the eigenvalues of the covariance matrix $X_N^*X_N$ for $N\to\infty$. We prove that the empirical density of eigenvalues in an interval $[E,E+\eta]$ converges to the Marchenko-Pastur law locally on the optimal scale, $N \eta /\sqrt{E} \gg (\log N)^b$, and in any interval up to the hard edge, $\frac{(\log N)^b}{N^2}\lesssim E \leq 4-\kappa$, for any $\kappa >0$. As a consequence, we show the complete delocalization of the eigenvectors.
\end{abstract}

\section{Introduction}
Let $X$ be a $ N\times M$ matrix with entries $x_{ij} = \Re x_{ij} + i \Im x_{ij}$. We assume that $\Re x_{ij}$ and $\Im x_{ij}$ are independent identically distributed real random variables with mean zero and variance $1/2$ so that
\beq
\E x_{ij} =0\quad\textrm{and}\quad \E|x_{ij}|^2 =1 \qquad i=1,\dots,N, \, j = 1, \dots, M\,,
\eeq
In what follows we shall denote by $X_N$ the scaled matrix
\begin{equation}
\label{e:XN}
X_N=X/\sqrt{N}.
\end{equation}
We denote by $\nu$ the probability distribution of $\Re x_{ij}$ and $\Im x_{ij}$. Let $s_\al$, $\al=1,...,N$, be the eigenvalues of $X_N^*X_N$. Since $X_N^*X_N$ is positive definite we can assume that $0 \leq s_1\leq s_2 \leq \dots \leq s_N$. The results of this paper extend easily to $X_N$ having real entries; to simplify the notation, we will consider in the following only the case of complex entries. 

Assume $d = \lim N/M > 0$ and let \[\lambda_{\pm} = \sqrt{1 \pm d^2}.\]
Marchenko and Pastur showed in \cite{MP67} the convergence of the density of the eigenvalues $s_1,\dots, s_N$ towards the Marchenko-Pastur law
\begin{equation}
\label{MP}
\rho_{MP}(E) = \frac{1}{2\pi}\sqrt{\frac{(\lambda_{+}-E)(E - \lambda_{-})}{E^2}},
\end{equation}
whenever $E \in [\lambda_{-}, \lambda_{+}]$ and 0 otherwise. In this paper, we will be interested in the case $d=1$. In this case the Marchenko-Pastur law is supported on the interval $[0,4]$ and is given by
\[ \rho_{MP}(E) = \frac{1}{2\pi} \sqrt{\frac{(4-E)}{E}} \]
It has therefore a $E^{-1/2}$ singularity close to the origin $E = 0$. This reflects the fact that the typical distance between eigenvalues is of order $\sqrt{E}/N$ rather than $1/N$, as it is in the bulk; for this reason, $E=0$ is known as the hard edge of the sample covariance matrix $X_N^* X_N$ (soft edges are instead characterized by the fact that the typical distance between neighbouring eigenvalues is larger than in the bulk). While the result of \cite{MP67} determines the convergence to (\ref{MP}) on intervals of order one, containing typically order $N$ eigenvalues, in the present paper we establish the convergence of the density of states locally, on intervals containing typically a bounded number of eigenvalues, independent of $N$. In particular, we consider intervals close to the hard edge $E=0$. As a direct consequence of the local validity of the Marchenko-Pastur law, we obtain the complete delocalization of the eigenvectors associated to eigenvalues up to the edge. A  further possible application of our results consists in establishing the universality of the local eigenvalue correlations close to the hard edge; this can be obtained following the receipt of \cite{eprsy}, making use of the result of \cite{ap05}, in the case of complex entries, or similarly to \cite{esy-inv,esyy12}, using the method of the local relaxation flow, for both $X_N$ having real or complex entries. We observe, however, that the universality of the local eigenvalue correlations close to the hard edge (where they can be described in terms of the so called Bessel kernel) has already been established, using a different approach, in \cite{TV0}. 

\bigskip

In the last years, a lot of progress was achieved in the spectral analysis of random matrices. Local convergence of the density of states of Wigner matrices to the semicircle law and delocalization of the eigenvectors has been established in \cite{esy-ap09,ESYcmp09,esy09,eyy10b}. Universality of the local eigenvalue correlations was proven for Wigner ensembles with arbitrary symmetry (real symmetric, hermitian, or quaternion hermitian ensembles) in \cite{esy-inv,eyy10b}. This result was obtained by the introduction of the local relaxation flow, a flow for the eigenvalues of the Wigner matrix with the property of fast relaxation to equilibrium (and such that, locally, it remains close to the Dyson Brownian motion described by the eigenvalue when the entries are evolved by independent Brownian motions). For ensembles of hermitian Wigner matrices, universality was proven earlier in \cite{eprsy,tv09,erstvy}. In all these proofs of universality, the local convergence of the density of states was a crucial ingredient. Universality at the edge of Wigner matrices was proven in \cite{S} and more recently in \cite{tv10cmp,abp,eyy10b,ly}. For sample covariance matrices with $0< d < 1$, local convergence to the Marchenko-Pastur law and universality of the local eigenvalue correlations were determined in the bulk \cite{esyy12,tv-rxv09} and at the soft edge \cite{WangSoftEdge,abp,py12}. More recently, local convergence of the density of states and delocalization results have also been obtained for more structured ensembles, such as the adjacency matrices of Erd{\H o}s-R{\'e}nyi graphs \cite{spa1, spa2} and band matrices \cite{ban}. In this paper, we focus on the hard edge of sample covariance matrices, proving the local convergence of the density of states to the Marchenko-Pastur law on the optimal scale (up to logarithmic corrections). As a consequence, we obtain complete delocalization of the eigenvectors associated with eigenvalues close to the hard edge. 

\bigskip

After the completion of our work, we learned that, independently from us, Bourgade, Yau and Yin study in \cite{byy} the convergence of the density of the eigenvalues of a random matrix $X$ with no symmetry constraints towards the circular law, on optimal scales. The basic ingredient of their proof is the study of the spectrum of the hermitization $(X-z)^* (X-z)$. In particular, for $z=0$, they obtain results similar to ours for the eigenvalues of sample covariance matrices.

\bigskip

An important object used in the proof of the local validity of the Marchenko-Pastur law is the Stieltjes transform defined for any $\theta \in \C$ with $\text{Im } \theta > 0$  by
\begin{equation}
\label{e:DeN}
\Delta_N (\theta)
=
\frac{1}{N} \Tr (X_N^*X_N -\theta )^{-1}
=
\frac{1}{N}\sum_{\alpha=1}^N \frac{1}{s_\al - \theta}\,.
\end{equation} 
In a similar way, one defines $\Delta (\theta)$ to be the Stieltjes transform of the Marchenko Pastur distribution. In the case $d=1$ that will be considered in this paper 
\begin{equation}
\label{e:De}
\Delta(\theta )= \int_{\RE} \frac{1}{x-\theta} \rho (x) dx = -\frac{1}{2} + \frac{1}{2} \sqrt{1-\frac{4}{\theta}}\,.
\end{equation}
Local convergence towards the Marchenko-Pastur law follows from the convergence of $\Delta_N$ towards $\Delta$. 

To simplify our analysis, we will assume that $\nu$ has subgaussian decay, i.e., that there exists $\de_0>0$ such that
\begin{equation}
\label{e:gd}
\int_\RE e^{\de_0 x^2} d\nu(x) < \infty\,.
\end{equation}
This condition is needed to apply a version of a theorem of Hanson and Wright as formulated in
\cite[Prop. 4.5]{esy09}, see also Proposition \ref{p:HW} below. At the price of getting weaker convergence rates, this assumption can be substantially relaxed (existence of sufficiently high moments is sufficient). Furthermore, we assume that probability density function of the real and imaginary parts of the entries is bounded; this will simplify the study of the eigenvalues located very close to the origin. Also in this case, improvements are certainly possible.

\bigskip  

Our first result is a proof of a bound on the number of eigenvalues $s_\alpha$ in a window $I=[E,E+\eta]$, valid up to the hard edge and for small $\eta$ (s.t. $N\eta/\sqrt{E} \gg (\log N)^b$, $b > 2$). The proof of the following theorem can be found in Section \ref{section apriori}.

\begin{theorem}
\label{t:apriori}
Let $X_N$ be a $N\times N $ matrix as described in \eqref{e:XN}, whose entries satisfy (\ref{e:gd}). 
Let $I=[E,E+\eta]$ with $N\eta/\sqrt{E} \geq (\log N)^b$, for some $b > 2$.  Denote by $\NH_I$ the number of eigenvalues of $X_N^*X_N$ in $I$. Then there exist constants $c,C, K_0>0$ such that, for any $K \geq K_0$ and $N$ large enough,
\begin{equation}
\label{e:apriori}
\PP\left(\NH_I\geq K\frac{N\eta}{\sqrt E}\right)
\leq    e^{-c \sqrt{K \frac{\eta N}{\sqrt{E}}} }
\,.
\end{equation}
\end{theorem}

Using the a priori bound in Theorem \ref{t:apriori} we prove the convergence of the Stieltjes transform $\Delta_N(E+i\eta)$ of the sample covariance matrices towards the Stieltjes trasnform $\Delta(E+i\eta )$ of the Marchenko-Pastur law, up to the hard edge, and close to the real axis. 

\begin{theorem}
\label{t:sti}
Let $X_N$ be a $N\times N $ matrix as described in \eqref{e:XN}, whose entries satisfy (\ref{e:gd}). Assume moreover that the probability density function of the real and imaginary part of the entries is bounded. Moreover set $\theta = E + i \eta$, with $E\leq 4-\kappa$, $0< \eta < \kappa E$, $N\eta / \sqrt{E} \geq (\log N)^b$, for some $b > 4$ and $0 < \kappa < 1$ (these bounds also imply that $E \geq (\log N)^{2b}/\kappa^2 N^2$). Then there exist $\ve_0 > 0$, $c>0$ such that for $0<\ve<\ve_0$ and $N$ large enough,
\beq
\PP\left( \left|\Delta_N(\theta)-\Delta(\theta)\right|\geq \frac{\ve}{\sqrt{E}}
\right) \leq e^{-c \ve \sqrt{\frac{N\eta}{\sqrt{E}}}} +  e^{-c(\log N)^{b/4}}  \,.
\eeq
\end{theorem}
The proof of this theorem is in Section \ref{section selfconsistent equation}. The convergence of the Stieltjes transform immediately implies the convergence of the density of states.

\begin{theorem}
\label{p:db}
Let $X_N$ be a $N\times N $ matrix as described in \eqref{e:XN}, whose entries satisfy (\ref{e:gd}). Assume moreover that the probability density function of the real and imaginary part of the entries is bounded. Suppose $E\leq 4-\kappa$, $0< \eta < \kappa E$, $N\eta / \sqrt{E} \geq (\log N)^b$, for some $b > 4$ and $0 < \kappa < 1$, and let $I = [E, E+\eta]$. 
Then  there exist $\ve_0 > 0$, $C, c>0$ such that for $0<\ve<\ve_0$ and $N$ large enough,
\begin{equation}
\label{e:db}
\PP\left(\left|\frac{\NN_{I}}{N \eta }- \frac{1}{\eta} \int_E^{E+\eta} ds  \rho(s) \right|
  \geq \frac{\ve}{\sqrt{E}}   \right) \leq C e^{-c \ve^2 \sqrt{\frac{N\eta}{\sqrt{E}}}} + C e^{c(\log N)^{b/4}}\,.
\end{equation}
\end{theorem}
The proof of Theorem \ref{p:db} can be obtained from Theorem \ref{t:sti}, similarly as in \cite[Cor. 4.2]{esy-ap09}. Finally, Theorem \ref{p:db} implies complete delocalization of the normalized eigenvectors of $X_N^*X_N$ associated with eigenvalues in the window $\left[\frac{(\ln N)^{b}}{\kappa^2 N^2},4-\kappa\right]$, for any $\kappa>0$. 

\begin{theorem}[Delocalization]
\label{t:del}
Let $X_N$ be a $N\times N $ matrix as described in \eqref{e:XN}, whose entries satisfy (\ref{e:gd}). Assume moreover that the probability density function of the real and imaginary part of the entries is bounded. Fix $0 < \kappa < 1$, $b > 4$. Then there exist constants  $c,\,C>0$ such that and for $N$ large enough,
\begin{equation}
\label{e:del}
\begin{aligned}
& \PP\left(\exists\, \uu \text{ s.t. } X_N^*X_N \uu = s \uu\,,\; \|\uu\|=1\,,\; s\in\left[\frac{(\ln N)^{2b}}{\kappa^2 N^2},4-\kappa\right]\;\textrm{and}\;  \|\uu\|_\infty \geq C \frac{(\ln N)^{\frac{b}{2}}}{N^{\frac{1}{2}}}\right)
\leq e^{-c (\ln N)^{\frac{b}{4}}}\,.
\end{aligned}
\end{equation}
\end{theorem}

\section{Basic definitions and results}
\label{s:def-res}

In this section we collect several definitions and results which will be used to prove the main theorems. 

\subsection{A formula for $\Delta_N$}

The proofs of Theorems \ref{t:apriori} and \ref{t:sti} rely on the following formula for the diagonal components of the resolvent $\left(X_N^*X_N-\theta\right)^{-1}$ (see \cite{esyy12})
\begin{equation}
\begin{split}
\label{e:resHkk}
\left((X_N^*X_N -\theta)^{-1}\right)_{kk} &=
\frac{1}{|\ww_k|^2-\theta - \ww_k^* W_k\left(W_k^*W_k - \theta\right)^{-1}W_k^*
\ww_k}\,\\
&= - 
\frac{1}{\theta \left(1+  \ww_k^* \left(W_kW_k^*-\theta\right)^{-1}
\ww_k \right)}\,.
\end{split}\end{equation}
where $\ww_k = \xx_k / \sqrt{N}$ is the $k$-th column  of the matrix $X_N$ and $W_k$ denotes 
the $N\times (N-1)$ matrix obtained by removing the $k$-th column from the matrix $X_N$ (notice that  $W_k^*W_k$ is the $(N-1)\times (N-1)$ minor of $X_N^*X_N$, obtained by removing
the $k$-th row and the $k$-th column). We used here the well-known identity \[ W_k\left(W_k^*W_k-\theta\right)^{-1}W_k^* = W_k^*W_k\left(W_kW_k^*-\theta\right)^{-1}\, \]
valid for $\text{Im } \theta \not = 0$, which can be proved using the Neumann expansion of the resolvent. Eq. (\ref{e:resHkk}) gives the following formula for the Stieltjes transform $\Delta_N(\theta)$:
\beq
\label{e:res2}
\Delta_N (\theta) = - \frac{1}{N} \sum_{k=1}^N
\frac{1}{\theta \left(1+  \ww_k^* \left(W_kW_k^*-\theta\right)^{-1}
\ww_k \right)}\,.
\eeq

\subsection{Properties of $\Delta$}

We collect here some properties of the Stieltjes transform $\Delta (\theta)$ of the Marchenko-Pastur distribution $\rho_{MP}$, defined by
\begin{equation}\label{eq:De} \Delta (\theta) = \int_\R dx \frac{1}{x-\theta} \rho_{MP} (x) = -\frac{1}{2} + \frac{1}{2} \sqrt{1-\frac{4}{\theta}} \end{equation}
where we use the branch of the square root with $\text{Re} \sqrt{1-4/\theta} \geq 0$. We use the fixed point equation
\begin{equation}
\label{eq:De-fix}
\theta (\Delta (\theta) +1) = -\frac{1}{\Delta (\theta)}.
\end{equation}

\begin{lemma}\label{l:propDe}
Let $\theta = E + i \eta$, with $\eta >0$. Then 
\begin{equation}\label{eq:De-b1} |\Delta (\theta)|^2 \leq \frac{1}{E} \quad \text{and }  \quad |1 + \Delta (\theta)|^2 \geq max \left\{ \frac{E}{E^2 + \eta^2} , \frac{1}{4} \right\} \end{equation}
Moreover 
\begin{equation}\label{eq:De-b2}  Im \, \Delta (\theta) \geq C \, \frac{|E-4|^{1/2} + \eta^{1/2}}{(E^2 + \eta^2)^{1/4}} \end{equation}
if $E^2 + \eta^2 \leq 4 E$ (this condition defines a circle of radius $2$ around $(E,\eta) = (2,0)$). 
%Finally, for $\theta = E + i\kappa E$, $E > 0$, $\kappa > 0$, we have
%\begin{equation}\label{eq:De-b3} Im \, \Delta (\theta) \geq \frac{\kappa^
%{1/2}}{E^{1/2}}
%\end{equation}
\end{lemma} 
\begin{proof}
{F}rom (\ref{eq:De-fix}), taking the  imaginary part,  we get
\begin{equation}
\label{e:imDelta}
\Im\Delta \left(1-|\Delta|^2 E\right) = \eta \Re (1+\Delta) |\Delta|^2\,.
\end{equation}
Eq. \eqref{eq:De} implies that $\Re(1+\Delta)>0$. Since $\Im \Delta (E+i\eta) >0$ for $\eta >0$,  together with \eqref{e:imDelta}, this implies the first bound in \eqref{eq:De-b1}. To get the second bound in \eqref{eq:De-b1} we first  notice that by (\ref{eq:De}), $\Re(1+\Delta) > 1/2$. This implies immediately that $|1+\Delta|^2>1/4$. The bound $|1+\Delta|^2> E/(E^2 +\eta^2)$ follows instead from (\ref{eq:De-fix}), combined with $|\Delta|^2 \leq 1/E$.

To show (\ref{eq:De-b2}), we observe that, from (\ref{eq:De}), 
\[ \text{Im } \Delta = \frac{1}{2} \text{Im } \sqrt{1-\frac{4}{\theta}} \geq \frac{1}{4} \sqrt{\left|1-\frac{4}{\theta}\right|} \geq  C \,  \frac{|E-4|^{1/2} + \eta^{1/2}}{(E^2 + \eta^2)^{1/4}} \]
under the assumption that $E^2 + \eta^2 < 4 E$. Here we used the fact that $\text{Im } \sqrt{z} \geq |z|^{1/2}/\sqrt{2}$, if $\text{Re } z \leq 0$ and $\text{Im } z \geq 0$. 

%Finally, Eq. (\ref{eq:De-b3}) follows from the observation that, with $
%\theta = E + i\eta$, $E,\eta > 0$, we have  
%\[ \begin{split} \text{Im } \Delta &= \frac{1}{2} \text{Im} \sqrt{ \frac
%{\theta-4}{\theta}} = \frac{1}{2} \frac{\text{Im } 
%\sqrt{E^2 + \eta^2 - 4 E + 4i\eta}}{(E^2 + \eta^2)^{1/2}} \\ &\geq \frac{2 %\eta}{(E^2 + \eta^2)^{1/2} ((E^2 + \eta^2 - 4 E)^2 + 16 \eta^2)^{1/4}} 
%\geq \frac{\eta^{1/2}}{E} \end{split} \]
%where we used the fact that $\text{Im } \sqrt{z} \geq \text{Im z}/|z|^
%{1/2}$, for $\text{Im } z \geq 0$ (which follows from the convexity bound %$\sin (\alpha/2) \geq (\sin \alpha)/2$, valid 
%for all $0 \leq \alpha \leq \pi$). 
\end{proof}

\subsection{Large deviations of quadratic forms} 

We will make use of the following inequality for the fluctuations of quadratic forms, due to Hanson and Wright. For the proof of the next proposition we refer to \cite[Prop. 4.5]{esy09}, see also \cite[App. B]{eyy10} and \cite{hw71}.

\begin{proposition}
\label{p:HW}
For $j=1,\dots ,N$ let $x_{j} = \text{Re } x_j + i \text{Im } x_j$, where $\{ \text{Re } x_{j} , \text{Im } x_j \}_{j=1}^N$ is a sequence of $2N$ real iid random variables, whose common distribution satisfies \eqref{e:gd}. Let $A = (a_{ij})$ be a $N\times N$ complex  matrix. Then there exist constants 
$c,C > 0$ such that, for any $\de >0$
\beq
\PP\left(\left|\sum_{i,j=1}^Na_{ij}\left(x_i\bar x_j - \E x_i \bar x_j\right) \right|\geq \de \right) \leq C e^{-c\min\{\de/\sqrt{\Tr A^* A},\de^2/\Tr A^* A\}}\,,
\eeq
\end{proposition}

The following proposition is a consequence of the Hanson-Wright inequality. Its proof can be found, for example, in \cite{esy09}. 

\begin{proposition}
\label{p:xi-1}
Let $\{\vv_\al\}_{\al=1}^m$ be  a set of $m$ orthonormal vectors in $\C^N$ and 
$\xx = (x_1, \dots , x_N)\in \C^N$, where $\{ \text{Re } x_j , \text{Im } x_j \}_{j=1}^N$ are $2N$ real iid random variables satisfying (\ref{e:gd}) such that $\E x_{j} = 0$, and $\E  |x_j|^2 = 1$. Then there exist two constants $c,C>0$ such that
\beq
\PP\left(
\sum_{\alpha=1}^m  |\xx \cdot \vv_\al|^2 \leq  \frac{m}{2}
\right) \leq C e^{-c \sqrt{m}}\,.
\eeq
\end{proposition}

\section{Upper bound for the number of eigenvalues: Proof of Theorem \ref{t:apriori}}
\label{section apriori}

Recall that $I = [E,E+\eta]$, and that $\NH_I$ denotes the number of eigenvalues of the matrix $X_N^* X_N$ in $I$. We have
\[ \begin{split} \NH_I &= \sum_{\alpha =1}^N {\bf 1} (s_\alpha \in [E, E+\eta]) \leq C \sum_{\alpha=1}^N \frac{\eta^2}{(s_\alpha-E)^2 + \eta^2}  = C \eta \text{Im } \sum_{\alpha = 1}^N \frac{1}{s_\alpha - E - i \eta} \\ &=
C \eta \text{Im } \Tr \, (X_N^* X_N - E - i \eta)^{-1} = C \eta \text{Im } \sum_{k=1}^N (X_N^*X_N - E - i \eta)^{-1}_{kk}  \\ &= - C \eta \text{Im } \sum_{k=1}^N \frac{1}{\theta \left(1+  \ww_k^* \left(W_kW_k^*-\theta\right)^{-1} \ww_k \right)}\,.
\end{split}\]
where we put $\theta = E + i\eta$ and we used \eqref{e:resHkk}. Using the spectral decomposition of $W_k W_k^*$, we find 
\[ \begin{split} \NH_I  &\leq - C \eta \text{Im } \sum_{k=1}^N  
\frac{1}{\theta - \sum_{\al = 0}^{N-1} |\ww_k \cdot \vv_\al^{(k)}|^2  + \sum_{\al=1}^{N-1}
  \frac{s_\al^{(k)} |\ww_k \cdot \vv_\al^{(k)}|^2}{\s_\al^{(k)}-\theta}} \\
 &\leq C \sum_{k=1}^N  
\frac{1}{1 + \sum_{\al=1}^{N-1}
  \frac{s_\al^{(k)} |\ww_k \cdot \vv_\al^{(k)}|^2}{(\s_\al^{(k)}-E)^2 + \eta^2}}\\
  &\leq C \frac{N^2 \eta^2}{E}  \frac{1}{N} \sum_{k=1}^N \frac{1}{\sum_{\al : s_\alpha^{(k)} \in [E, E+\eta]} N |\ww_k \cdot \vv_\al^{(k)}|^2}
   \end{split} \]
where, in the second inequality, we used the fact that $|\text{Im } (1/z)| \leq 1/|\text{Im } z|$. Setting $K = 2 C^{1/2}$, it follows that
\[ \NH_I \leq K \frac{N\eta}{\sqrt{E}} \]
unless there exists $k \in \{1, \dots , N\}$, with
\[  \sum_{\al : s_\alpha^{(k)} \in [E, E+\eta]} N |\ww_k \cdot \vv_\al^{(k)}|^2 \leq  \frac{\NH_I}{4} \]
In other words,
\[\begin{split} \PP & \left(\NH_I \geq K \frac{N\eta}{\sqrt{E}} \right) \\ & \; = \PP \left(\NH_I \geq K \frac{N\eta}{\sqrt{E}} \text{ and $\exists k \in \{ 1, \dots, N\}$ with} \sum_{\al : s_\alpha^{(k)} \in [E, E+\eta]} N |\ww_k \cdot \vv_\al^{(k)}|^2 \leq  \frac{\NH_I}{4} \right) \end{split} \]
Since $W_k^* W_k$ is a minor of $X_N^* X_N$, it follows that its eigenvalues are interlaced between the eigenvalues of $X_N^* X_N$. This implies that 
\[ \left| \{ \al : s^{(k)}_\al  \in [E , E+\eta] \} \right| \geq \NH_I - 1 \geq \frac{\NH_I}{2} \] 
on the event we consider. Proposition \ref{p:xi-1} (applied with $E=1,\eta=0$) implies therefore that
\[ \PP  \left(\NH_I \geq K \frac{N\eta}{\sqrt{E}} \right) \leq C N e^{-c \sqrt{ \frac{K N\eta}{\sqrt{E}}}} \leq C e^{-c \sqrt{ \frac{K N\eta}{\sqrt{E}}}}  \]
after adjusting the constants. This concludes the proof of Theorem \ref{t:apriori}. 

\section{An estimate for the number of eigenvalues close to zero}
\label{s:wegner}

In this section, we show that, with high probability, there cannot be too many eigenvalues at distances smaller than $1/N^2$ from $0$. To this end, we need the boundedness of the probability density function of the entries. We use here the notation $\NH [a,b]$ to indicate the number of eigenvalues in  $[a,b]$. 

\begin{proposition}
\label{p:wegner}
Let $X_N=(x_{ij}/\sqrt{N})$ be a $N\times N$ matrix as described in equation \eqref{e:XN}. Assume that the probability density function $h(x)$ of $\text{Re } x_{ij}$ and $\text{Im } x_{ij}$ is bounded. Then there exists a constant $C , c >0$ such that 
\begin{equation}\label{e:wegner} \PP \left( \NH \left[0,\frac{K}{N^2}\right] \geq L\right) \le C e^{- L} \end{equation} 
for all $L > c K$. 
\end{proposition}
\begin{proof}
We start with the observation that
\[ 
\PP (\NH [0, K/N^2] \geq L ) = \PP ( (\NH [0,K/N^2] / L)^p \geq 1) \leq  \E (\NH [0,K/N^2] / L)^p \]
Next, we notice that 
\[ \begin{split}  \NH [0, K/N^2] &\leq \frac{CK}{N^2}  \sum_{k=1}^N \text{Im} \left(X_N^* X_N - \theta \right)^{-1}_{kk} 
\end{split} \]
where we set $\theta = K N^{-2} + i K N^{-2}$. This implies that
\[ \left(\frac{\NH [0, K / N^2]}{L} \right)^p \leq \left( \frac{CK}{NL} \right)^p \left( \frac{1}{N} \sum_{k=1}^N  
\text{Im} \left(X_N^* X_N - \theta \right)^{-1}_{kk}  \right)^p \leq  \left( \frac{CK}{NL} \right)^p \left( \frac{1}{N} \sum_{k=1}^N  
(\text{Im} \left(X_N^* X_N - \theta \right)^{-1}_{kk})^p  \right) \]
by H\"older inequality. Hence
\[\begin{split} \PP  (\NH [0, K/N^2] \geq L ) &\leq  \left( \frac{CK}{NL} \right)^p  \, \E 
( \text{Im} \left(X_N^* X_N - \theta \right)^{-1}_{11})^p \\
&\leq 
 \left( \frac{CK}{NL} \right)^p \E \, \left| \frac{1}{\theta + \theta \sum_{\al=0}^{N-1}
  \frac{|\ww_1 \cdot \vv_\al^{(1)}|^2}{\s_\al^{(1)} - \theta}}  \right|^p \\ &\leq 
   \left( \frac{CK}{L} \right)^p  \E  \frac{1}{\left(\sum_{\alpha=0}^{N-1} c_\alpha
|\xx_1 \cdot \vv_\al^{(1)}|^2\right)^p}
 \end{split} \]
where we neglected the real part of the denominator and we defined $\xx_1 
= \sqrt{N} \ww_1$ (it is a vector in $\C^N$, whose components have iid real and imaginary parts with zero mean and variance $1/2$), and 
\begin{align}
c_\alpha &= \frac{1}{( N^2 s_\alpha^{(1)}/K - 1)^2+1}
\label{e:yal}
\end{align}
We recall that the eigenvalues   $s_{\al}$  and  $s_{\al}^{(1)}$  are ordered in increasing order. 
On the event $\NH [0,K/N^2] \geq L$, the interlacing property implies that at least $L-1$ eigenvalues of the minor are in the interval $[0,K/N^2]$; i.e. $\s_{\al}^{(1)} \in [0,K/N^2]$  for $\al = 0,. .. , L-1$.  This implies that $c_\alpha \geq 1/2$ for all $\al = 1, \dots , L-1$. Therefore
\[  \PP  (\NH [0, K/N^2] \geq L ) \leq \left( \frac{2CK}{L} \right)^p  \E  \frac{1}{\left(\sum_{\alpha=0}^{L-1} |\xx_1 \cdot \vv_\al^{(1)}|^2 \right)^p} \]
Next, we take $p= (L-2)/2$. Since the matrix entries are assumed to have a bounded probability density function, Lemma A.1 of \cite{ms11} implies that 
\beq
\E_{\xx_1}  \frac{1}{\left(\sum_{k = 1}^{L-1}  |\xx_1 \cdot \vv_\al^{(1)}|^2 \right)^{(L-2)/2}}  \leq C \]
for $C>0$ indpendent of $L$. This concludes the proof of the proposition.
%(it actually proves even more, that the probability decays as $L^{-L/2}$, but we will not make use of %this fact).
\end{proof}

\section{Convergence of the Stieltjes transform: Proof of Theorem \ref{t:sti}}
\label{section selfconsistent equation}

We start from the formula (\ref{e:res2}), rewritten as
\begin{equation}\label{eq:DeN-1}
\Delta_N = -\frac{1}{N} \sum_{k=1}^N \frac{1}{\theta \left(1+ \Delta_N + \frac{\Omega_k}{\sqrt{E}} \right)} \end{equation}
which immediately implies that
\begin{equation}\label{eq:DeN-err} \Delta_N + \frac{1}{\theta (\Delta_N +1)} = - \frac{1}{N} \sum_{k=1}^N \frac{\Omega_k/\sqrt{E}}{\theta (1+ \Delta_N + \Omega_k/ \sqrt{E})(1+\Delta_N)} \end{equation}
Here we defined the error terms
\[ \begin{split} \Omega_k = \; &\sqrt{E} \left( \ww_k^* (W_k W_k^* - \theta)^{-1} \ww_k - \frac{1}{N} \Tr \, (W_k W_k^* - \theta)^{-1} \right) \\ &+\sqrt{E} \left( \frac{1}{N} \, Tr \, (W_k W_k^* -\theta)^{-1} - \frac{1}{N} \Tr (X_N^* X_N - \theta)^{-1} \right) \end{split} \]
Observe that, with probability one, 
\begin{equation}
\label{e:tr1}
\sqrt{E} \left|\frac{1}{N}\Tr
(X_N^*X_N - \theta)^{-1}-\frac{1}{N}\Tr(W_k W_k^* -\theta)^{-1}\right|\leq
  C\frac{\sqrt{E}}{N\hscale}\,.
\end{equation}
This follows from 
\[ \frac{1}{N}\Tr(W_kW_k^*-\theta)^{-1} = - \frac{1}{N\theta} + \frac{1}{N} \Tr (W_k^* W_k - \theta)^{-1} \]
and from the interlacing of the eigenvalues of $W_k^* W_k$ between the eigenvalues of $X_N^* X_N$. To estimate the first difference in the error $\Omega_k$, we notice that 
\[ \E_{\ww_k} \ww_k^* (W_k W_k^* - \theta)^{-1} \ww_k = \frac{1}{N} \Tr \, (W_k W_k^* - \theta)^{-1} \]
Therefore, defining the matrix $A = (a_{ij})$, with
\[ a_{ij} = \frac{\sqrt{E}}{N} \sum_{\alpha} \frac{\overline{\vv_\alpha} (i) \vv_\alpha (j)}{s_\alpha^{(k)} - \theta} \]
and the vector $\xx = \sqrt{N} \ww_k = (x_1, \dots , x_N)$ (this is a vector in $\C^N$, whose components are order one random variables), we find
\[  \sqrt{E} \left(\ww_k^* (W_k W_k^* - \theta)^{-1} \ww_k - \frac{1}{N} \Tr \, (W_k W_k^* - \theta)^{-1} \right) = \sum_{ij} a_{ij} (x_i \overline{x}_j - \E_{\xx} x_i \overline{x}_j) \]
and therefore (taking into account also (\ref{e:tr1}))
\[ \PP ( |\Omega_k| \geq \eps) \leq \PP \left( \left| \sum_{ij} a_{ij} (x_i \overline{x}_j - \E_{\xx} x_i \overline{x}_j) \right| \geq \eps \right) \]
Observe that
\[ \Tr A^* A = \frac{E}{N^2} \sum_\alpha \frac{1}{|s_\alpha^{(k)} - \theta|^2} \]
In Lemma \ref{p:Atilde}, we show that, up to an event with probability at most $e^{-c (\log N)^{b/4}}$, 
\[ \Tr A^* A \leq C \frac{\sqrt{E}}{N\eta} \]
Therefore, Proposition \ref{p:HW} implies that
\[ \PP  ( |\Omega_k| \geq \eps) \leq C e^{-c (\log N)^{b/4}} + C e^{-c \eps \sqrt{\frac{N\eta}{\sqrt{E}}}} \]
and thus
\begin{equation}\label{eq:PPmaxX} \PP \left( \max_{k=1, \dots , N} |\Omega_k| \geq \eps \right) \leq C e^{-c (\log N)^{b/4}} + C e^{-c \eps \sqrt{\frac{N\eta}{\sqrt{E}}}} \end{equation}
after adjusting the constants. We restrict now our attention to the event $|\Omega_k| \leq \eps$ for all $k=1, \dots , N$.

%To complete the proof of Theorem \ref{t:sti}, we use a continuity 
%argument. 
%We fix $\kappa > 0$, and consider $\theta = E + i \eta \in \C$ with $0 < E %< 4-\kappa$, $0 < \eta < \kappa E$, $N \eta / \sqrt{E} \geq (\log N)^{b/4}
%$. We connect $\theta$ with a point $\theta_0 = E_0 + i \wt{\kappa} E_0$, %for large $E_0 > 0$. To this end, we define the two segments $L_1 = \{ (E, %\wt{\eta}) : \eta < \wt{\eta} \leq \wt{\kappa} E \}$ where $\wt{\kappa} > %\kappa$ is chosen s.t. $E^2 (1+ \wt{\kappa}^2) < 4 E$, and $L_2 = \{ (\wt
%{E}, \wt{\kappa} \wt{E}) : E \leq \wt{E} \leq E_0 \}$. We let $L = L_1 
%\cup L_2$. It follows from Lemma \ref{l:propDe} that, on $L$, $\text{Im } %\Delta > C_0 / \sqrt{E}$ for a constant $C_0$ depending only on $\kappa$. 

To complete the proof of Theorem \ref{t:sti}, we use a continuity argument. 
We fix $\kappa > 0$, and consider $\theta = E + i \eta \in \C$ with $0 < E < 4-\kappa$, $0 < \eta < \kappa E$, $N \eta / \sqrt{E} \geq (\log N)^{b}$. We connect $\theta$ with the point $\theta_0 = 2 + 2 i \kappa$. Let $L$ denote the line segment connecting $\theta$ and $\theta_0$. Note that, on $L$, $(E-2)^2 + \eta^2 \leq 4$ always holds. Hence, Lemma \ref{l:propDe} implies that, on $L$, $\text{Im } \Delta > C_0 / \sqrt{E}$ for a constant $C_0$ depending only on $\kappa$ ($C_0$ can be chosen as $\text{const} \cdot \kappa^{1/2}/(1+\kappa^2)^{1/4}$).

We claim now that, if $|\Delta_N - \Delta| \leq C_0/  (2 \sqrt{E})$ somewhere on $L$, then $|\Delta_N - \Delta| \leq C \eps/\sqrt{E}$ where $C$ depends only on $\kappa$. In fact, $\text{Im } \Delta > C_0 / \sqrt{E}$ and $|\Delta_N - \Delta| \leq C_0/ (2\sqrt{E})$ imply that $\text{Im } \Delta_N > C_0 / (2\sqrt{E})$. This implies (for $\eps < C_0/4$) that $\text{Im } (\Delta_N + (\Omega_k/\sqrt{E})) \geq C_0 / (4\sqrt{E})$. Hence (\ref{eq:DeN-err}) gives
\[ \left| \Delta_N + \frac{1}{\theta (\Delta_N + 1)} \right| \leq \frac{8\eps}{C_0^2 \sqrt{E}} \]
Subtracting the fixed point equation $\Delta + 1/ (\theta (\Delta+1)) = 0$, we find
\[ \left| (\Delta_N - \Delta) \left( 1 + \frac{\Delta}{\Delta_N+1} \right) \right| \leq \frac{8\eps}{C_0^2 \sqrt{E}}  \]
Therefore
\[ \begin{split} | \Delta_N - \Delta| \leq \; & \frac{8\eps}{C_0^2 \sqrt{E}} \frac{|\Delta_N + 1|}{|\Delta_N + \Delta + 1|} \\ \leq 
\; & \frac{8\eps}{C_0^2 \sqrt{E}}  \left[ {\bf 1} (|\Delta_N + 1| > 2 |\Delta|) \frac{|\Delta_N + 1|}{|\Delta_N + \Delta + 1|}  + {\bf 1} 
 (|\Delta_N + 1| < 2 |\Delta|) \frac{|\Delta_N + 1|}{|\Delta_N + \Delta + 1|} \right] \\ \leq 
\; & \frac{16\eps}{C_0^2 \sqrt{E}} \left(1+\frac{1}{C_0}\right)  
 \end{split}
\]
where we used that, on $L$, $\text{Im } \Delta > C_0/\sqrt{E}$ and $|\Delta| \leq 1/ \sqrt{E}$ (see Lemma \ref{l:propDe}). Theorem \ref{t:sti} follows because, from \cite{MP67}, $|\Delta_N (2 + 2 i\kappa) - \Delta (2 + 2 i \kappa)| \leq C_0 / (2\sqrt{2})$ for $N$ large enough.
This completes the proof of Theorem \ref{t:sti}.

\begin{lemma}\label{p:Atilde}
Let $X_N = (x_{ij} / \sqrt{N})$ be an $N \times M$ matrix as defined in (\ref{e:XN}). Denote by $s_\al$ the eigenvalues of $X_N^* X_N$. Assume the real and imaginary part of the entries $x_{ij}$ are iid random variables with a common bounded probability density function. Let $\theta  = E + i \eta$, with $\eta \leq E$ and $N\eta/ \sqrt{E} \geq (\log N)^b$. Then there exist constants $C_1, C_2, c >0$ with 
\begin{equation}
\label{e:boundAtilde}
\PP \left(\frac{E}{N^2}\sum_{\al=1}^{N-1}\frac{1}{|s_\al-\theta|^2} \geq C_1 \frac{\sqrt E}{N\eta}\right) \leq C_2 e^{-c \left(\log N\right)^{b/4}}.
\end{equation}
\end{lemma}
\begin{proof}
We clearly have
\[ \begin{split} 
\PP \left( \frac{E}{N^2}\sum_{\al=1}^{N-1}\frac{1}{|s_\al-\theta|^2} \geq \; C_1 \frac{\sqrt E}{N\eta}\right) \leq \; &\PP \left( \frac{E}{N^2}\sum_{\al: s_\al \leq (\log N)^b/N^2} \frac{1}{|s_\al-\theta|^2} \geq \frac{C_1}{3}  \frac{\sqrt E}{N\eta}\right) \\ &+ \PP \left( \frac{E}{N^2}\sum_{\al: s_\al \in [(\log N)^b/N^2, E/2]} \frac{1}{|s_\al-\theta|^2} \geq \frac{C_1}{3}  \frac{\sqrt E}{N\eta}\right) \\ &+ \PP 
\left( \frac{E}{N^2}\sum_{\al: s_\al > E/2} \frac{1}{|s_\al-\theta|^2} \geq \frac{C_1}{3}  \frac{\sqrt E}{N\eta}\right) =  \text{I} + \text{II} + \text{III} \end{split} \]
We start by controlling the term $\text{I}$. To this end, note that
\[ \frac{E}{N^2}\sum_{\al: s_\al \leq (\log N)^b/N^2} \frac{1}{|s_\al-\theta|^2} \leq \frac{E}{N^2 \eta^2} \, \NH [0, (\log N)^b / N^2] \]
where, as usual, $\NH[a,b]$ denotes the number of eigenvalues of $X_N^* X_N$ in the interval $[a,b]$. Hence, by Prop. \ref{p:wegner},
\[ \text{I} \leq \PP \left( \NH [0, (\log N)^b/ N^2] \geq \frac{C_1}{3} \frac{N\eta}{\sqrt{E}} \right) \leq C e^{- c \frac{N\eta}{\sqrt{E}}} \leq C e^{-c (\log N)^b} \]
for appropriate constants $C,c>0$. Next, we consider the term $\text{II}$. We have 
\[ \begin{split} 
\frac{E}{N^2}\sum_{\al: s_\al \in [(\log N)^b / N^2 , E/2]} \frac{1}{|s_\al-\theta|^2} \leq \frac{E}{N^2} \sum_{k=2}^{k_0} \frac{\NH [ 2^{-k} E ,  2^{-k+1} E]}{E^2} 
\end{split} 
\]
where $k_0 > 0$ is chosen as the smallest integer with $2^{-k_0} E \leq (\log N)^b / N^2$. From 
Theorem \ref{t:apriori}, it follows that, for a sufficiently large $K >0$, 
\[ \NH [ 2^{-k} E ,  2^{-k+1} E] \leq C N 2^{-k/2}  \sqrt{E} \]
up to an event of probability at most $\exp (-c (2^{-k/2}N\sqrt{E})^{1/2})$. This implies that, apart from an event of total probability bounded by
\[ \sum_{k=2}^{k_0} e^{-c \sqrt{2^{-k/2} N\sqrt{E}}} \leq k_0 \, e^{-c \sqrt{2^{-k_0/2} N \sqrt{E}}} \leq k_0 \, e^{-c (\log N)^{b/4}}  \leq e^{-(c/2) (\log N)^{b/4}} \]
we have
\[  \frac{E}{N^2}\sum_{\al: s_\al \in [(\log N)^b / N^2 , E/2]} \frac{1}{|s_\al-\theta|^2} \leq C \frac{E}{N^2} \sum_{k=2}^{k_0} \frac{2^{-k/2} N \sqrt{E}}{E^2} \leq C \frac{1}{N\sqrt{E}} \leq C \frac{\sqrt{E}}{N\eta} \]
where we used the assumption $\eta \leq E$. 
Finally, we have to control the term $\text{III}$. To this end, we observe that
\[ \frac{E}{N^2} \sum_{\al : s_\al > E/2} \frac{1}{|s_\alpha - \theta|^2} \leq 
\frac{E}{N^2} \sum_{\al : s_\al > E/2} \frac{1}{(s_\al - E)^2 + \eta^2} \leq \frac{E}{N^2} \sum_{k \geq 0}  \frac{\NH_{I_k}}{2^{2k} \eta^2} \]
where we set $I_k = J_k \cap [ E/2, \infty)$, with $J_k = [E-2^k \eta, E-2^{k-1} \eta] \cup [ E+ 2^{k-1} \eta , E+ 2^k \eta]$ for $k \geq 1$ and $J_0 = [E-\eta, E+\eta]$. Observe that, by Theorem \ref{t:apriori},
\[ \NH_{I_k} \leq C \frac{2^k N\eta}{\sqrt{E}} \]
up to an event with probability at most
\[ e^{-c \sqrt{ \frac{2^k N \eta}{\sqrt{E}}}} \leq e^{-c \sqrt{ \frac{N\eta}{\sqrt{E}}}} \leq e^{-c (\log N)^{b/4}} \]
Therefore,
\[ \frac{E}{N^2} \sum_{\al : s_\al > E/2} \frac{1}{|s_\al - \theta|^2} \leq \frac{E}{N^2} \sum_{k \geq 0}  \frac{2^k N \eta}{2^{2k} \eta^2 \sqrt{E}} \leq C \frac{\sqrt{E}}{N\eta} \]
apart from an event with probability at most $e^{-c (\log N)^{b/4}}$. This completes the proof of the lemma.
\end{proof}

\section{Delocalization: Proof of Theorem \ref{t:del}}
\label{s:del}

Denote by  $\{s_\al\}_{\al=1}^N$  the eigenvalues of the matrix $X_N^*X_N$ with $s_1\leq ... \leq s_N$ and by $\{\uu_\al\}_{\al=1}^N$ the corresponding set of orthonormal eigenvalues.  From the equation $X_N^*X_N \uu_\al = s_\al \uu_\al$ and the condition $\|\uu_\al\|^2 = 1$ it follows that (see also  \cite[Cor. 25]{tv-rxv09} and \cite{ESYcmp09, esy-ap09})
\begin{equation*}
| \uu_\al (k)|^2 = \frac{1}{1+\frac{1}{N}\sum_{\beta=1}^{N-1} \frac{s_\beta^{(k)} |\uu_\beta^{(k)} \cdot \xx_k|^2}{(s_\alpha-s_\beta^{(k)})^2}} \end{equation*}
where $\xx_k = \sqrt{N} \ww_k$ and $\ww_k$ denotes the $k$-th column of the matrix $X_N$, while 
$s_\beta^{(k)}$ and $\uu_\beta^{(k)}$ are the eigenvalues and the corresponding eigenvectors of 
the matrix $W_k^*W_k$, where the matrix $W_k$ is obtained by removing the $k$-th column from the matrix $X_N$. For arbitrary $\eta >0$, we have 
\[ 
| \uu_\al (k)|^2 \leq \frac{N\eta^2}{s_\alpha \sum_{\beta: s_\beta^{(k)} \in [s_\alpha, s_\alpha + \eta]} |\uu_\beta^{(k)} \cdot \xx_k|^2}\]
Taking $\eta = \sqrt{s} \frac{(\log N)^b}{N}$, Theorem \ref{p:db} implies that
\[ \left|\{ s^{(k)}_\beta \in [s_\alpha, s_\alpha + \eta] \} \right| \geq C (\log N)^b \]
up to an event with probability smaller than $e^{-c (\log N)^b}$. Prop. \ref{p:xi-1} implies then that
\[ \sum_{\beta : s_\beta^{(k)} \in [s_\alpha, s_\alpha+ \eta]} |\uu_\beta^{(k)} \cdot \xx_k|^2 \geq \frac{(\log N)^b}{2} \] 
apart from an event with probability smaller than $e^{-c (\log N)^b}$. This implies that
\[ \PP \left(|\uu_\alpha (k)|^2 \geq \frac{(\log N)^b}{N}\right) \leq C e^{-c (\log N)^b} \]
Taking the maximum over $k$, Theorem \ref{t:del} follows.

\end{document}